\documentclass{amsart}
\usepackage{amssymb, amscd}
\usepackage{enumerate}
\usepackage[dvips,dvipdf]{graphicx}
\usepackage[usenames,dvipsnames]{color}

\newcommand\datver[1]{\def\datverp
 {\par\boxed{\boxed{\text{Version: #1; Run: \today}}}}}
\datver{2.2; Revised: 14 July, 2011}
\newcommand{\ie}{{\em i.e., }}
\newcommand{\diam}{\operatorname{diam}}

\newcommand\vt[1]{ \mathbf{#1} }

\newcommand{\maC}{\mathcal C}
\newcommand{\maD}{\mathcal D}

\newcommand{\maK}{\mathcal K}
\newcommand{\maL}{\mathcal L}

\newcommand{\maO}{\mathcal O}

\newcommand{\maR}{\mathcal R}
\newcommand{\maS}{\mathcal S}
\newcommand{\maT}{\mathcal T}
\newcommand{\maV}{\mathcal V}

\newcommand{\CC}{\mathbb C}
\newcommand{\NN}{\mathbb N}
\newcommand{\PP}{\mathbb P}
\newcommand{\RR}{\mathbb R}

\newcommand{\TT}{\mathbb T}
\newcommand{\ZZ}{\mathbb Z}

\newcommand\Hk{H_{\vt k}}
\newcommand\Hkn{H_{\vt k, n}}

\newcommand\TmS{\TT \smallsetminus \maS}

\newtheorem{theorem}{Theorem}[section]

\newtheorem{lemma}[theorem]{Lemma}

\theoremstyle{definition}

\theoremstyle{remark}
\newtheorem{remark}[theorem]{Remark}

\newtheorem{algorithm}[theorem]{Algorithm}

\author[E. Hunsicker]{Eugenie Hunsicker} \address{Eugenie Hunsicker,
  Department of Mathematical Sciences, Loughborough University,
  Loughborough, Leicestershire, LE11 3TU, UK }
\email{E.Hunsicker@lboro.ac.uk}

\author[H. Li]{Hengguang Li} \address{Hengguang Li, Department of
  Mathematics, Wayne State University, Detroit, MI 48202, USA}
\email{hli@math.wayne.edu}

\author[V. Nistor]{Victor Nistor} \address{V. Nistor, Pennsylvania
  State University, Math. Dept., University Park, PA 16802, USA, and
  Inst. Math. Romanian Acad.  PO BOX 1-764, 014700 Bucharest Romania}
\email{nistor@math.psu.edu}

\author[V. Uski]{Ville Uski} \address{Ville Uski, Department of
  Mathematical Sciences, Loughborough University, Loughborough,
  Leicestershire, LE11 3TU, UK } \email{V.Uski@lboro.ac.uk}

\date\today \thanks{V.N. was partially supported by the NSF Grants
   OCI-0749202 and DMS-1016556. Manuscripts available
  from {\bf
    http:{\scriptsize//}www.math.psu.edu{\scriptsize/}nistor{\scriptsize/}}.
  H.L. was partially supported by the NSF Grant DMS-1115714.  V.U. was
  supported by and EH was supported in part by Leverhulme Trust
  grant J11695.}

\begin{document}

\title[Schr\"odinger operators]{Analysis of Schr\"odinger operators
  with inverse square potentials {II}: FEM and approximation of
  eigenfunctions in the periodic case}

\date{\today}


\begin{abstract}
Let $V$ be a {\em periodic} potential on $\RR^3$ that is smooth
everywhere except at a discrete set $\maS$ of points, where it has
singularities of the form $Z/\rho^2$, with $\rho(x) = |x - p|$ for $x$
close to $p$ and $Z$ is continuous, $Z(p) > -1/4$ for $p \in \maS$. We
also assume that $\rho$ and $Z$ are smooth outside $\maS$ and $Z$ is
smooth in polar coordinates around each singular point. Let us denote
by $\Lambda$ the periodicity lattice and set $\TT := \RR^3/
\Lambda$. In the first paper of this series \cite{HLNU1}, we
obtained regularity results in weighted Sobolev space for the
eigenfunctions of the Schr\"odinger-type operator $H = -\Delta + V$
acting on $L^2(\TT)$, as well as for the induced $\vt k$--Hamiltonians
$\Hk$ obtained by resticting the action of $H$ to Bloch waves. In this
paper we present two related applications: one to the Finite Element
approximation of the solution of $(L+\Hk) v = f$ and one to the
numerical approximation of the eigenvalues, $\lambda$, and
eigenfunctions, $u$, of $\Hk$. We give optimal, higher order
convergence results for approximation spaces defined piecewise
polynomials. Our numerical tests are in good agreement with the
theoretical results.
\end{abstract}

\maketitle

\tableofcontents

\section{Introduction and statement of main results}

In this paper, which is the second part of work begun in \cite{HLNU1},
we present applications of the theoretical regularity results in the
first part of this paper to Finite Element Method approximation
schemes.  The first application is to approximation of eigenvalues,
$\lambda$, and eigenfunctions, $u$, of the Bloch operator, $\Hk$,
associated to a periodic Hamiltonian operator with inverse square
potential at isolated points.  For example, one of our main results,
Theorem \ref{athm1.eig} yields optimal orders of convergence for the
Finite Element approximations of the eigenvalues of $\Hk$ using graded
meshes.  These rates are higher than those that can be obtained using
standard meshes.  The second application is to the Finite Element
Method, again using graded meshes, applied to equations of the form
$(L + \Hk)v = f$.  The final section of this paper presents numerical
tests showing good agreement with our theoretical results for this
second problem.

Hamiltonian operators with inverse square potentials arise in a
variety of interesting contexts.  The standard example of a
Schr\"odinger operator with $c/\rho$ potential is a special case of the
inverse square potentials we consider, where
the function $\rho^2V$ vanishes to order 1 at the singularity, 
and the results of this work apply to such
operators.  But in addition, Hamiltonians with true inverse square
potentials arise in relativistic quantum mechanics from the square of
the Dirac operator coupled with an interaction potential, and they
arise in the interaction of a polar molecule with an electron.  See
\cite{MorozSchmidt, Sprung} for further applications of inverse square
potentials to physics. See also \cite{Bespalov, MR2421883, Li09, LN09, Westphal}
for related results on operators with singular coefficients.  Thus it
is interesting in several areas of physics to understand how to
approximate solutions to equations involving such operators.

Before we can state our approximation results, we must fix some
notation and state the assumptions we make about our Hamiltonian
operators.  Consider a Hamiltonian operator $H := -\Delta + V$ that is
periodic on $\mathbb{R}^3$ with triclinic periodicity lattice
$\Lambda$.  Its fundamental domain is a parallelopiped whose faces can
be identified under the symmetries of $H$ to form the torus
$\mathbb{T} = \mathbb{R}^3/\Lambda$, which is how we will denote this
fundamental domain in the remainder of this paper.  Let $\rho(x)$ be a
continuous function on $\mathbb{T}$ that is given by $\rho(x) =|x-p|$
for $x$ close to $p$, is smooth except at the points of $\maS$, and
may be assumed to be equal to one outside a neighbourhood of $\maS$.

We need two assumptions about the potentials $V$ that we will consider
in this paper.  First, we assume that $V$ is smooth except at a set of
points $\maS \subset \mathbb{T}$, near which it has singularities of
the form $Z/\rho^2$, where $Z$ is continuous on $\mathbb{T}$ and
smooth in polar coordinates around $p$.  We denote this as follows.
\begin{equation}\label{eq.def.Z}
  \text{\bf Assumption 1}: \qquad\ Z := \rho^2 V \in \maC(\TT) \cap
  \maC^\infty(\overline{\TmS}).
\end{equation}
Assumption 1, more precisely the continuity of $Z$ at $\maS$, allows
us to formulate our second assumption. Namely,
\begin{equation}\label{eq.Z1/4}
\text{\bf Assumption 2}: \eta:= \min_{p \in \maS} \sqrt{1/4 + Z(p)} > 0.
\end{equation}
In particular, we assume that for all $p \in \maS$, $Z(p) > -1/4$.
These assumptions are sharp in the sense that the analysis yields
fundamentally different results if either one fails. In particular,
the value $\eta = -1/4$ corresponds to the critical coupling for an
isolated inverse square potential in $\mathbb{R}^3$ where the system
undergoes a transition between the conformal and non-conformal regimes
\cite{MorozSchmidt}.  If the first assumption fails, then the
available analytic techniques are much weaker, see for instance
\cite{Felli1, Felli2}. In either case, the approximation theorems in
this paper fail if either assumption is violated.  More details of
this are included in the first part, \cite{HLNU1}, and a study of the
analysis when these assumptions are relaxed will be examined in a
forthcoming paper.

We are interested in understanding the spectrum and generalised
eigenfunctions of the operator $H$.  As usual, we do this by studying
Bloch waves.  Recall that if $\vt k$ is an element of the first
Brillion zone of $\Lambda$, that is, is an element of the fundamental
domain of the dual lattice of $\Lambda$, then a Bloch wave with wave
vector $\vt k$ is a function in $L^2_{loc}(\mathbb{R}^3)$ that
satisfies the semi-periodicity condition
\begin{equation}\label{eq.Block}
  \psi_{\vt k}(x+X) = e^{i{\vt k}\cdot X} \psi_{\vt k}(x) \qquad
  \forall X \in \Lambda.
\end{equation}
It is well known that such a Bloch wave can be written as
\begin{equation}\label{eq.Block2}
  \psi_{\vt k}(x) = e^{i{\vt k}\cdot x} u_{\vt k} (x)
\end{equation}
for a function $u_{\vt k}$ that is truly periodic with respect to
$\Lambda$ and thus can be considered as living on the three-torus
$\TT$.  We define the $\vt k$--Hamiltonian $\Hk$ on $L^2(\TT)$ by
\begin{equation}\label{eq:hamk}
  \Hk := - \sum_{j=1}^3 (\partial_j + i\rm{k}_j)^2 + V.
\end{equation}
Then we have further that if a Bloch wave $\psi_{\vt k}$ is a
generalized eigenfunction of $H$ with generalized eigenvalue
$\lambda$, then the function $u_{\vt k} := e^{ -i{\vt k}\cdot x}
\psi_{\vt k}(x)$ is a standard $L^2$-eigenfunction of $\Hk$ with
eigenvalue $\lambda$.  Let $\lambda_{j}$, $j \ge 1$, be the
eigenvalues of $\Hk$, arranged in increasing order, $\ldots \le
\lambda_{j} \le \lambda_{j+1} \le \ldots$, and repeated according to
their multiplicities. That is, if $E(\lambda)$ denotes the eigenspace
of $\Hk$ corresponding to $\lambda$, then $\lambda$ is repeated $\dim
(E(\lambda))$ times.
  
As usual, for our finite element approximation results, we consider a
sequence $S_n$ of finite dimensional subspaces of the domain of $\Hk$
and let $R_n$ denote the Riesz projection onto $S_{n}$, that is, the
projection in the bilinear form $((L+\Hk) y,w)_{L^2(\TmS)}$, (for a
suitable $C \ge 0$), associated to $\Hk$. Let $\Hkn := R_n \Hk R_n$ be
the associated finite element approximation of $\Hk$, acting on
$S_n$. Denote by $\lambda_{j,n}$ the eigenvalues of the approximation
$\Hkn$, again arranged in increasing order, $\ldots \le \lambda_{j, n}
\le \lambda_{j+1, n} \le \ldots$, and repeated according to their
multiplicities and let $u_{j, n} \in S_n$ be a choice of corresponding
eigenfunctions (linearly independent). The spaces $S_n$ we use for our
theorems are defined in terms of a sequence of graded tetrahedral
meshes $\maT_n:= k^n(\maT_0)$ on $\TT$ (sometimes called
triangulations), given by sequential refinements, associated to a
scaling parameter $k$, of an original tetrahedral mesh $\maT_0$.  We
describe the meshing refinement procedure in detail in Section
\ref{sec3}.  We will take $S_n = S(\maT_n, m),$ the finite element
spaces associated to these meshes (\ie using continuous, piecewise
polynomials of degree $m$).

Our first theorem, which is a theoretical result for the finite
element method approximation of eigenvalues and eigenfunctions of
$\Hk$ using tetrahedralisations with graded meshes, is as follows.

\begin{theorem}\label{athm1.eig}
Let $\lambda_j$ be an eigenvalue of $\Hk$ and fix $0<a<\eta$, $a \leq
m$. Let $\lambda_{j,n}$ be the finite element approximations of
$\lambda_j$ associated to the nested sequence $\maT_n$ of meshes on
$\TT$ defined by the scaling parameter $k = 2^{-m/a}$ {and
  piecewise polynomials of degree $m$. Also, let $u_{j, n}$ be an
  eigenbasis corresponding to $\lambda_{j,n}$.} Then there exists a
constant $c(\lambda_j,a)$ independent of $n$ such that the following
inequalities hold for a suitable eigenvector $u_j \in E(\lambda_j)$:
\begin{equation*}
  |\lambda_j -  \lambda_{j, n}| \le c(\lambda_j,a) \dim(S_n)^{-2m/3},
\end{equation*}
\begin{equation*}
	\|u_{j} - u_{j, n} \|_{\maK^1_1(\TmS)} \leq c(\lambda_j,a)
        \dim(S_n)^{-m/3}, 
\end{equation*}
where the space $\maK^1_1(\TmS)$ is a weighted Sobolev space defined
below in Equation \ref{eq.def.ws}.
\end{theorem}

For our second theorem, we consider the finite element approximations
of the equation
\begin{equation}\label{eq.PDE.eq}
  (L + \Hk)v = f, \quad \mbox{ for } L > C_0,
\end{equation}
where $C_0$ is the constant from Theorem \ref{theorem1} below. We then
define the form $a(y, w) := ((L + \Hk) y, w)$ and let $v$ be the
solution of Equation \eqref{eq.PDE.eq} above. We then define the usual
Galerkin Finite Element approximation $v_n$ of $v$ as the unique $v_n
\in S_n := S(\maT_n, m)$ such that
\begin{equation}\label{eq.FEM.eq}
  a(v_n, w_n) := \big( (L + \Hk)v_n, w_n \big) = (f, w), \mbox{
    for all } w_n \in S_n.
\end{equation}

Then Theorem \ref{athm1.eig}together with the Lax-Milgram Lemma and
Cea's lemma imply that we have the following $h^m$ quasi-optimal rate
of convergence.

\begin{theorem}\label{athm1.fem}
The sequence $\maT_n := k^n(\maT_0)$ of meshes on $\PP$ defined using
the $k$-refinement, for $k =2^{-m/a}$, $0<a<\eta$, $a \le m$, and
piecewise polynomials of degree $m$, has the following property. The
sequence $v_{n} \in S_n := S(\maT_n, m)$ of Finite Element (Galerkin)
approximations of $v$ from Equation \eqref{eq.FEM.eq} satisfies
\begin{equation}\label{nnn1}
	\|v - v_{n} \|_{\maK^1_1(\TmS)} \leq C \dim(S_n)^{-m/3}
        \|f\|_{\maK^{m-1}_{a-1}(\TmS)},
\end{equation}
where $C$ is independent of $n$ and $f$. 
\end{theorem}

These theorems are interesting because it is known that the
convergence rate of a standard finite element method (\ie based on
quasi-uniform meshes) is limited. However, under the assumptions on
our potentials, if we use graded meshes instead, we can obtain an
approximation rate as fast as we like by using polynomials of
sufficiently high degree in the elements. This is due to the fact that
although regularity of the associated Bloch waves is limited in terms
of standard Sobolev spaces on $\TT$, it is arbitrarily good with
respect to weighted Sobolev spaces.  We will recall the definition of
these spaces and the relevant regularity results from \cite{HLNU1}
along with additional background in Section \ref{sec:background}.

The remainder of the paper is organised as follows. In Section
\ref{sec3}, we first describe the $k$-refinement algorithm for the
three dimensional tetrahedral meshes, which results in a sequence of
meshes $\maT_n$.  We then prove a general interpolation approximation
result for the sequence of finite element spaces associated to this
sequence of meshes. In Section \ref{sec.FEM} we use this general
approximation result to prove our main approximation results This
section includes in particular the proofs of Theorem \ref{athm1.eig}
and Theorem \ref{athm1.fem}, as well as an additional result about the
condition number of the stiffness matrix associated to the finite
element spaces, $S_n$.  In the last section, Section \ref{sec.tests},
we discuss results of numerical tests of the method for solving
equations of the form $(L+\Hk)v=f$ and compare them to the theoretical
results.

\subsection*{Acknowledgements} 
We would like to thank Bernd Ammann, Douglas Arnold, and Catarina
Carvalho for useful discussions. We also thank the Leverhulme Trust
whose funding supported the fourth author during this project. This
project was started while Hunsicker and Nistor were visiting the Max
Planck Institute for Mathematics in Bonn, Germany, and we are greatful
for its support.

\section{Background results}\label{sec:background}

In this section we recall some definitions and results from
\cite{HLNU1}, as well as the classical approximation result for
Lagrange interpolants (see \cite{BabuAziz, BrennerScott, Ciarlet91,
  SchwabBook}), that will be used in the proofs of the approximation
theorems above.  First, these results are given in terms of weighted
Sobolev spaces which are defined as follows:
\begin{equation}\label{eq.def.ws}
  \maK^m_a(\TmS ) := \{v : \TT \smallsetminus
  \maS \to \CC, \ \rho^{|\beta|-a} \partial^\beta v \in L^2(\TT),
  \ \forall\ |\beta| \leq m\}.
\end{equation}
These spaces have been considered in many other papers, most notably
in Kondratiev's groundbreaking paper \cite{kondratiev67}. 

The first result that we recall guarantees the existence of
solutions of equations of the form $(L + \Hk)v=f$ for $L$ greater than
some constant $C_0$, and identifies the natural domain of $\Hk$. Let
us fix smooth functions $\chi_p$ supported near points of $\maS$ such
that the functions $\chi_p$ have disjoint supports and $\chi_p = 1$ in
a small neighbourhood of $p \in \maS$. Then Theorem 1.1 and
Proposition 3.6 from \cite{HLNU1} combine to give right away the
following result.

\begin{theorem}\label{theorem1}
Let $V$ be a potential satisfying both Assumptions 1 and 2. Then there
exists $C_0 > 0$ such that $L + \Hk : \maK_{a+1}^{m+1}(\TmS) \to
\maK_{a-1}^{m-1}(\TmS)$ is an isomorphism for all $m \in \ZZ_{\ge 0}$,
all $|a| < \eta$, and all $L> C_0$. Moreover, for any $u \in
\maK_{a+1}^{m+1}(\TmS)$ satisfying $(L + \Hk)v = f \in H^{m-1}(\TmS)$,
we can find constants $a_p \in \RR$ such that
\begin{equation*}
    u_{reg} := u - \sum_{p \in \maS} \chi_p \rho^{\sqrt{1/4 + Z(p)} -
      1/2} \in \maK^{m+1}_{2}(\TmS ).
\end{equation*}
\end{theorem}

We obtain, in particular, that $\Hk$ has a natural self-adjoint
extension, the Friedrichs extension. Therefore, from now on, we shall
extend $\Hk$ to the domain of the Friedrichs extension of $L + \Hk$,
as in the above Theorem. Let us denote by $\maD(\Hk)$ its
domain. Then Theorem \ref{theorem1} gives that $\maD(\Hk) =
  \maK_{2}^{2}(\TmS)$ for $\min_{p} Z(p) > 3/4$, and, in general,
\begin{equation}\label{eq.domain.Hk}
   \maD(\Hk) \subset \maK_{a+1}^{2}(\TmS), \quad \mbox{for } a < \eta
   := \min_{p} \sqrt{1/4 + Z(p)} \mbox{ and } a \le 1
\end{equation}
so that $\maD(\Hk) \subset \maK_{1}^1(\TmS) \subset H^1(\TmS)$, since
we assumed that $\min_{p} Z(p) > -1/4$.

We can now state a regularity theorem for the eigenfunctions of $\Hk$
near a point $p \in \maS$, or equivalently, for Bloch waves associated
to the wavevector $\vt k$.

\begin{theorem}\label{theorem2}
Assume that $V$ satisfies Assumptions 1 and 2 and let $u \in
\maD(\Hk)$ satisfy $\Hk u = \lambda u$, for some $\lambda \in
\RR$. Then we can find constants $a_p \in \RR$ such that
\begin{equation*}
    u - \sum_{p \in \maS} \chi_p \rho^{\sqrt{1/4 + Z(p)} - 1/2} \in
    \maK^{m+1}_{a' +1}(\TmS ), \quad \forall a' < \min_{p \in \maS}
    \sqrt{9/4 + Z(p)}\,.
\end{equation*}
In particular, $u \in \maK^{m+1}_{a + 1}(\TmS),$ where $a < \eta :=
\min_{p \in \maS} \sqrt{1/4 + Z(p)}$ and $m \in \ZZ_+$ is arbitrary.
\end{theorem}

See also \cite{Kato51, Kato57} for some related classical results in
this area. Theorems \ref{theorem1} and \ref{theorem2} lead to an
estimate for the distance from an element in the domain of $\Hk$ to
the approximation spaces that we construct using graded meshes.

Next, recall the definition of Lagrange interpolants associated to a
mesh.  Let us choose $\PP$ to be a parallelopiped that is a
fundamental domain of the Lattice $\Lambda$. That is, $\RR^3 = \cup_{y
  \in \Lambda} (y + \overline{\PP})$ and all $y + \PP$ disjoint. Let
$\maT=\{T_i\}$ be a {\em mesh} on $\PP$, that is a mesh of $\PP$ with
tetrahedra $T_i$. We can identify this $\maT$ with a mesh $\maT'$ of
the fundamental region of the lattice $\maL$ (that is, to the
Brillouin zone of $\maL$). Fix an integer $m \in \NN$ that will play
the role of the order of approximation.  We denote by $S(\maT, m)$ the
finite element space associated to the degree $m$ Lagrange
tetrahedron. That is, $S(\maT, m)$ consists of all continuous
functions $\chi : \overline{\PP} \to \RR$ such that $\chi$ coincides
with a polynomial of degree $\leq m$ on each tetrahedron $T \in \maT$
and {\em $\chi$ is periodic}.  This means the values of $\chi$ on
corresponding faces coincide, so $\chi$ will have a continuous,
periodic extension to the whole space, or alternatively, can be
thought of as a continuous function on $\maT$. We shall denote by $w_I
= w_{I, \maT} \in S(\maT, m)$ the Lagrange interpolant of $w \in
H^2(\RR^3)$. Let us recall the definition of $w_{I, \maT}$.  First,
given a tetrahedron $T$, let $[t_0, t_1, t_2, t_3]$ be the barycentric
coordinates on $T$. The nodes of the degree $m$ Lagrange tetrahedron
$T$ are the points of $T$ whose barycentric coordinates $[t_0, t_1,
  t_2, t_3]$ satisfy $m t_j \in \ZZ$.  The {\em degree $m$ Lagrange
  interpolant} $w_{I, \maT}$ of $u$ is the unique function $w_{I,
  \maT} \in S(\maT, m)$ such that $w = w_{I, \maT}$ at the nodes of
each tetrahedron $T \in \maT$. The shorter notation $w_I$ will be used
when only one mesh is understood in the discussion.

The classical approximation result for Lagrange interpolants
 (\cite{BabuAziz, BrennerScott, Ciarlet91, SchwabBook}) can now be stated.
\begin{theorem}\label{athm1.classical}
Let $\maT$ be a mesh of a polyhedral domain $\PP \subset \RR^3$ with
the property that all tetrahedra comprising $\maT$ have angles $\geq
\alpha$ and edges $\leq h$.  Then there exists a constant $C(\alpha,
m)>0$ such that, for any $u\in H^{m+1}(\PP)$,
\begin{equation*}
	 \| u - u_I \|_{H^1(\PP)} \leq C(\alpha, m) h^m
         \|u\|_{H^{m+1}(\PP)}.
\end{equation*}
\end{theorem}

Finally, we recall two properties of functions in the weighted
Sobolev spaces $\maK^m_a(\TmS)$ that are useful for the analysis of 
the approximation scheme we use with graded meshes. 
The proofs of these lemmas are contained
in \cite{HunsickerNistorSofo} and are based on the definitions and
straightforward calculations.

\begin{lemma} \label{claim2}  
Let $D$ be a small neighborhood of a point $p\in \maS$ such that on
$D$, $\rho$ is given by distance to $p$.  Let $0<\gamma<1$ and denote
by $\gamma D$ the region obtained by radially shrinking around $p$ by
a factor of $\gamma$.  Then
\begin{equation*}
	\|w\|_{\maK^m_a(D)} = (\gamma)^{a-3/2}\|w\|_{ \maK^m_a(\gamma D) }.
\end{equation*}
\end{lemma}

\begin{lemma} \label{claim3}
If $m \geq m^{\prime}$, $a \geq a^{\prime}$ and $0 <\rho<\delta$ on
$D$, then
\begin{equation*}
    \|w\|_{\maK_{a^{\prime}}^{m^{\prime}}(D)} \leq 
    \delta^{a -a^{\prime}} \|w\|_{\maK^m_a(D)}.
\end{equation*}
\end{lemma}

We can now continue to the definition of the mesh refinement technique and the proof
of the general approximation theorem underlying our two main theorems.

\section{Approximation and mesh refinement\label{sec3}}

Our two main theorems follow from standard results, such as Cea's
Lemma (for the proof of Theorem \ref{athm1.fem}) and the results used in
\cite{BabuOsborn1, BabuOsborn2, BabuOsborn3, BrambleOsborn, Osborn}
(for the proof of Theorem \ref{athm1.eig}), together with the following
underlying approximation theorem:

\begin{theorem}\label{athm1.gen}
There exists a sequence $\maT_n$ of meshes of $\TT$ that depends only
on the choice of a parameter $k \le 2^{-m/a}$, $0 < a <\eta$ and
$a\leq m$, with the following property. If $u \in
\maK_{a+1}^{m+1}(\TmS)$, then the modified Lagrange interpolant $u_{I,
  \maT_n} \in S(\maT_n, m)$ of $u$ satisfies
\begin{equation*}
	\|u - u_{I, \maT_n} \|_{\maK^1_1 (\TmS )} \leq C
        \dim(S_n)^{-m/3} \|u \|_{\maK_{a+1}^{m+1}(\TmS)},
\end{equation*}
where $C$ depends only on $m$ and $a$ (so it is independent of $n$ and
$u$).
\end{theorem}

In this section we will define the mesh refinement process and prove
Theorem \ref{athm1.gen}.
The first step is to describe the refinement procedure that results
in our sequence of meshes (or triangulations).  This is based on the
construction in \cite{BNZ3D1} and in \cite{bey95}, thus we refer
the reader to those papers for details, and here give only an outline
and state the critical properties.  The second step is to prove a
sequence of simple lemmas used in the estimates.  The third step is to
prove the estimate separately on smaller regions. This uses the
scaling properties of the meshes in Lemmas \ref{claim2} and \ref{claim3} 
together with Theorem \ref{athm1.classical}.

\subsection{Construction of the meshes}
We continue to keep the approximation degree $m$ fixed throughout this
section.  Fix a parameter $a$ and let $k = 2^{-m/a}$. In
our estimates, we will chose $a$ such that $ a < \eta := \min_{p}
\sqrt{1/4 + Z(p)}$ and $a\leq m$. Let $l$ denote the smallest distance
between the points in $\maS$.  Choose an initial mesh $\maT_0$ of
$\PP$ with tetrahedra such that all singular points of $V$ (\ie all
points of $\maS$) are among the vertices of $\maT_0$ and no
tetrahedron has more than one vertex in $\maS$. We assume that this
mesh is such that if $F_1$ and $F_2$ are two opposite faces of $\PP$,
which hence correspond to each other through periodicity, then the
resulting triangulations of $F_1$ and $F_2$ will also correspond to
each other, that is, they are congruent in an obvious sense.

We start with a special refinement of an arbitrary tetrahedron $T$
that has one of the vertices in the set $\maS$. Our assumptions then
guarantee that all the other vertices of $T$ will not be in
$\maS$. Motivated by the refinement in \cite{bey95, BNZ3D1, MR2824854, MR2522959, LMN10}, we define
our {\em $k$-refinement algorithm for a single tetrahedron} that
divides $T$ into eight sub-tetrahedra as follows.

\begin{algorithm}\label{def.n1} {\bf $k$-refinement for a single
tetrahedron}:\ Let $x_0, x_1, x_2, x_3$ be the four vertices of $T$.

We denote $T$ by its vertex set $\{x_0, x_1, x_2, x_3\}$. Suppose that
$x_0 \in \maS$, so that $x_0$ is the one and only vertex that will be
refined with a ratio $k\in (0, 1/2]$. We first generate new nodes
$x_{ij}$, $0\leq i<j\leq 3$, on each edge of $T$, such that
$x_{ij}=(x_i+x_j)/2$ for $1\leq i< j\leq 3$ and
$x_{0j}=(1-k)x_0+k x_j$ for $1\leq j\leq 3$. Note that the
node $x_{ij}$ is on the edge connecting $x_{i}$ and $x_j$. Connecting
these nodes $x_{ij}$ on all the faces, we obtain 4 sub-tetrahedra and
one octahedron. The octahedron then is cut into four tetrahedra using
$x_{13}$ as the common vertex. Therefore, after one refinement, we
obtain eight sub-tetrahedra (Figure \ref{fig.tetrarefine}), namely,
\begin{eqnarray*}
  &\{x_0, x_{01}, x_{02}, x_{03}\}, \ \{x_1, x_{01}, x_{12},
  x_{13}\},\ \{x_2, x_{02}, x_{12}, x_{23}\},\ \{x_3, x_{03}, x_{13},
  x_{23}\}\\
  &\{x_{01}, x_{02}, x_{03}, x_{13}\}, \ \{x_{01}, x_{02}, x_{12},
  x_{13}\},\ \{x_{02}, x_{03}, x_{13}, x_{23}\},\ \{x_{02}, x_{12},
  x_{13}, x_{23}\}.
\end{eqnarray*}
\end{algorithm}

\begin{figure}
\centerline{\includegraphics[scale=0.26]{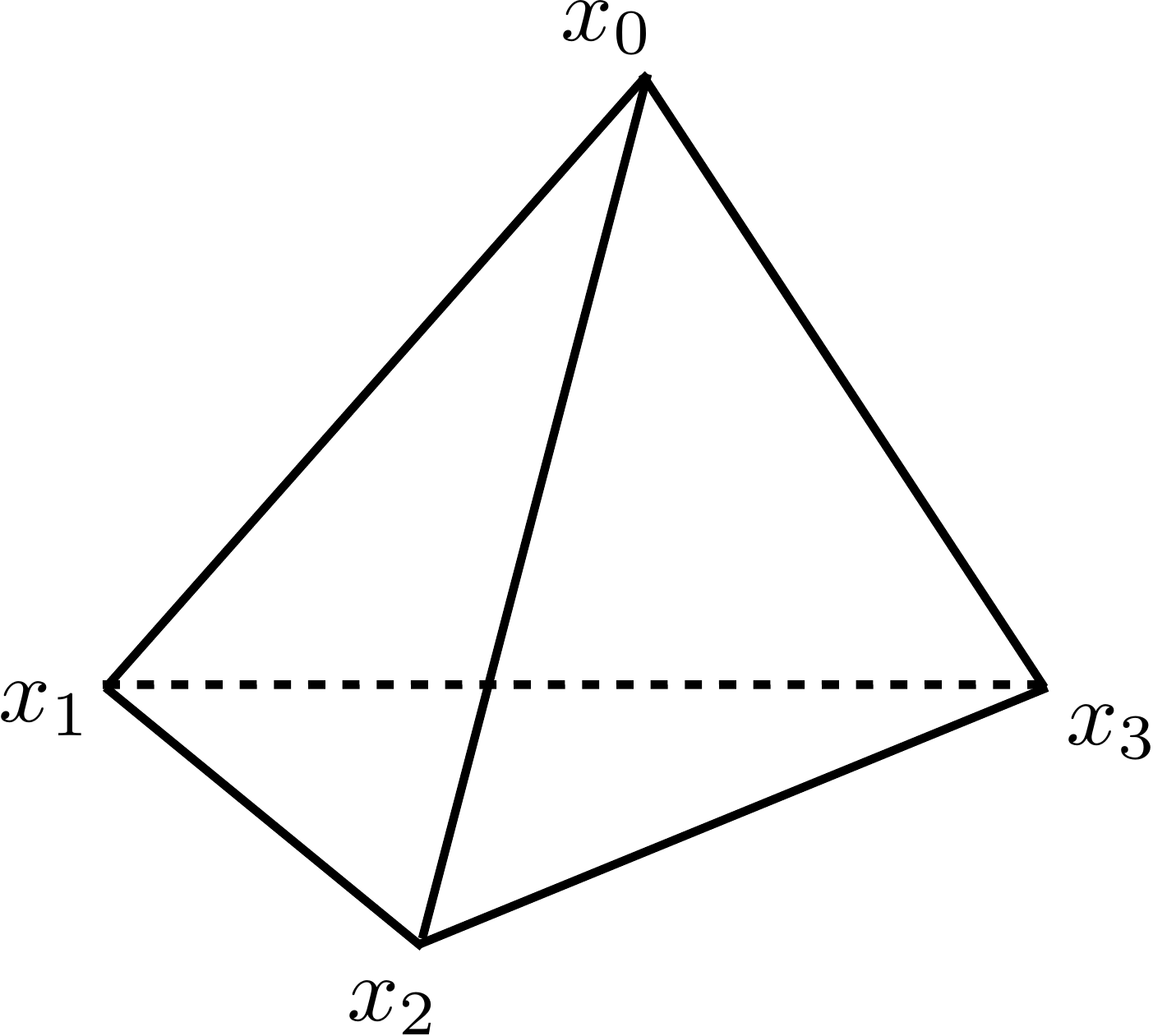} 
\hspace{3cm}\includegraphics[scale=0.26]{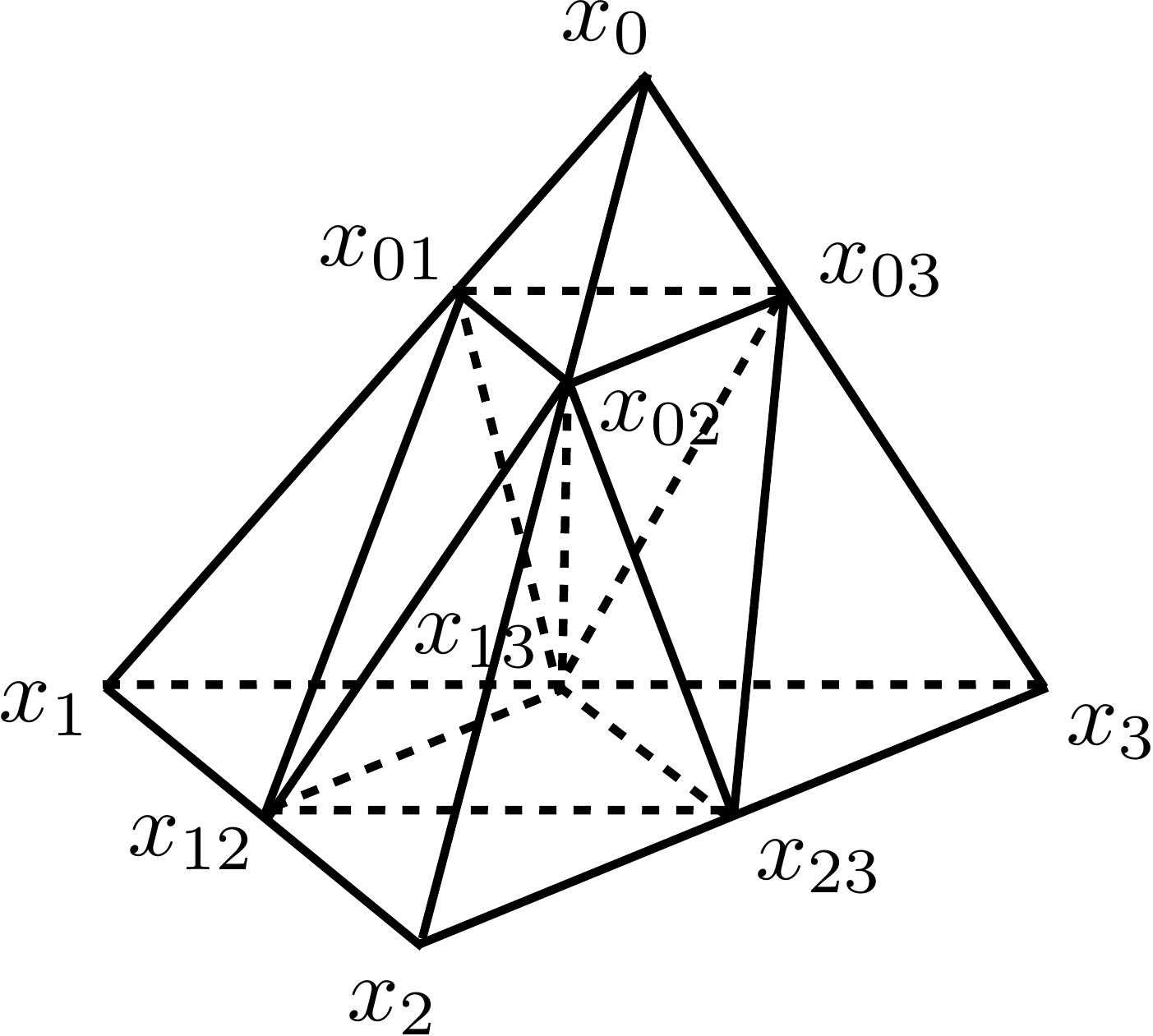}}
\caption{The initial tetrahedron $\{x_0, x_1, x_2, x_3\}$ (left);
  eight sub-tetrahedra after one $k$-refinement (right),
  $k=\frac{|x_0x_{01}|}{|x_0x_1|} = \frac{|x_0x_{02}|}{|x_0x_2|}
  = \frac{|x_0x_{03}|}{|x_0x_3|}$.}
\label{fig.tetrarefine}
\end{figure}

\begin{algorithm} \label{def.n2} {\bf $k$-refinement for a mesh}:\
Let $\maT$ be a triangulation of the domain $\PP$ such that all points
in $\maS$ are among the vertices of $\maT$ and no tetrahedron contains
more than one point in $\maS$ among its vertices. Then we divide each
tetrahedron $T$ of $\maT$ that has a vertex in $\maS$ using the
$k$-refinement and we divide each tetrahedron in $T$ that has no
vertices in $\maS$ using the $1/2$-refinement. The resulting mesh will
be denoted $k(\maT)$. We then define $\maT_n = k^n(\maT_0)$, where
$\maT_0$ is the intial mesh of $\PP$.
\end{algorithm}

\begin{remark}
According to \cite{bey95}, when $k=1/2$, which is the case when
the tetrahedron under consideration is away from $\maS$, the recursive
application of Algorithm \ref{def.n1} on the tetrahedron generates
tetrahedra within at most three similarity classes. On the other hand,
if $k <1/2$, the eight sub-tetrahedra of $T$ are not necessarily
similar. Thus, with one $k$-refinement, the sub-tetrahedra of $T$
may belong to at most eight similarity classes. Note that the first
sub-tetrahedra in Algorithm \ref{def.n1} is similar to the original
tetrahedron $T$ with the vertex $x_0\in \maS$ and therefore, a further
$k$-refinement on this sub-tetrahedron will generate eight
children tetrahedra within the same eight similarity classes as
sub-tetrahedra of $T$. Hence, successive $k$-refinements of a
tetrahedron $T$ in the initial triangulation $\maT_0$ will generate
tetrahedra within at most three similarity classes if $T$ has no
vertex in $\maS$. On the other hand, successive $k$-refinements
of a tetrahedron $T$ in the initial triangulation will generate
tetrahedra within at most $1+7\times 3=22$ similarity classes if $T$
has a point in $\maS$ as a vertex. Thus, our $k$-refinement is
conforming and yields only non-degenerate tetrahedra, all of which will
belong to only finitely many similarity classes.
\end{remark}

\begin{remark}
Recall that our initial mesh $\maT_0$ has matching restrictions to
corresponding faces. Since the singular points in $\maS$ are not on
the boundary of $\PP$, the refinement on opposite boundary faces of
$\PP$ is obtained by the usual mid-point decomposition. Therefore, the
same matching property will be inherited by $\maT_n$. In particular,
we can extend $\maT_n$ to a mesh in the whole space by periodicity. We
will, however, not make use of this periodic mesh on the whole space.
\end{remark}

For each point $p \in \maS$ and each $j$, we denote by $\maV_{pj}$ the
union of all tetrahedra of $\maT_j$ that have $p$ as a vertex.  Thus
$\maV_{pj}$ is obtained by scaling the tetrahedra in $\maV_{p0}$ by a
factor of $k^j$ with center $p$. In particular, the level $n \ge j$
refinements of $\maT_0$ give rise to a mesh on $\maR_{pj}:=
\maV_{p(j-1)} \smallsetminus \maV_{pj}$. Define
\begin{equation*}
  \Omega := \PP \smallsetminus \cup_{p \in \maS} \maV_{p0}.
\end{equation*}
According to Definition \ref{def.n2}, both $\Omega$ and $\cup_{p \in
  \maS} \maV_{p0}$ are triangulated using the $k$-refinement.  For
each tetrahedron in $\maV_{p0}$ we use the $k$-refinement for a single
tetrahedron, while for $\Omega$ we use the $1/2$--refinement for
meshes, which is, of course, a uniform refinement. Then, we can
decompose $\PP$ as the union
\begin{equation}\label{eq.decomp}
	\PP = \Omega \cup_{p \in \maS} \Big( \cup_{j=1}^{n} \maR_{pj}
        \cup \maV_{pn} \Big),
\end{equation}
where each set in the union is a union of tetrahedra in
$\maT_n$.

\begin{remark}\label{rk.45}
Note that the size of each simplex of $\maT_n$ contained in $\Omega$
is $\maO( 2^{-n})$, the size of each simplex of $\maT_n$ contained in
$\maR_{pj}$ is $\maO(k^{j} 2^{-(n-j)})$, and the size of
$\maV_{pn}$ is $\maO(k^{n})$. In addition, the number of
tetrahedra in $\maT_n$ is $\maO(2^{3n})$ (see Algorithm
\ref{def.n2}).
\end{remark}

We now define the finite element approximation $u_n\in S(\maT_n, m)$
to the equation $(L + \Hk)v=f$, where $\maT_n$ is obtained by
applying $n$ times the $k$-refinements to $\maT_0$, where
$k=2^{-m/a}$, $0< a< \eta$, $a\leq m$, and $\lambda
>0$ satisfies Theorem \ref{theorem1}. Then $u_n$ is defined for any
$v_n\in S(\maT_n, m)$ by
\begin{equation}\label{eqn.new8} 
  (\Hk u_n, v_n) + \lambda (u_n, v_n) := ( \nabla u_n, \nabla
  v_n)_{L^2} + (( V + \lambda)u_n, v_n)_{L^2} = (f, v_n)_{L^2}.
 \end{equation}
Note that Theorem \ref{theorem1} gives that the finite element
solution $u_n\in S(\maT_n, m)\subset \maK^1_1$ is well defined by
\eqref{eqn.new8}. The approximation properties of $u_n$ are discussed
in Theorem \ref{athm1.eig}.

\subsection{Proof of Theorem \ref{athm1.gen}} 
Note that the singular expansion of Theorem \ref{theorem2} shows that
the value of an eigenvalue $u$ of $\Hk$ at a singular point in $\maS$ may not be
defined. Therefore, we must define the \textit{modified} degree $m$
Lagrange ``interpolant'' $u_{I, n} = u_{I, \maT_n}$ associated to
the mesh $\maT_n$, such that
\begin{equation}\label{eqn.new7}
  \left\{\begin{array}{ll} u_{I,n}(x) = u(x)\ {\rm{ for\ any\ node
  }}\ x \notin\maS\\
  u_{I,n}(x)=0\ {\rm{if}}\ x\in\maS.\\
  \end{array}\right.
\end{equation}
Alternatively, we can take the modified Lagrange interpolant to be
zero on the whole tetrahedron that contains a singular point. By
construction, the restriction of $\maT_n$ to $\maR_{pj}$ scales to the
restriction of $\maT_{n-j+1}$ to $\maR_{p1}$. From now on, we refer to
$u_{I, n}=u_{I, \maT_n}$ as the \textit{modified} interpolation
defined in \eqref{eqn.new7}. The following lemma is based on the
definition of the $k$-refinement and the discussion in Remark
\ref{rk.45}.

\begin{lemma} \label{claim1}  
For all $x \in \maR_{pj}$, $u_{I, n}(x) = u_{I, n-j+1}(
k^{-(j-1)}(x) ),$ where $k^{-(j-1)}(x):= p +
(x-p)/k^{(j-1)}$ is the dilation with ratio $k^{-(j-1)}$ and
center $p$.
\end{lemma}

Recall that $\rho^2 V \in C^\infty(\overline{\TmS}) \cap \maC(\TT)$
and $\min_{p} Z(p) > -1/4$. That is, $V$ satisfies Assumptions 1 and 2.

We can now give the proof of Theorem \ref{athm1.gen}.

\begin{proof}
Recall that $\maV_{p0}$ consists of the tetrahedra of the initial mesh
$\maT_0$ that have $p$ as a vertex and that all the regions $\maV_p$
are away from each other (they are closed and disjoint).  We used this
to define $\Omega := \PP \smallsetminus \cup_p \maV_{p0}$. The region
$\maV_{pj}$ is obtained by dilating $\maV_{p}$ with the ratio
$k^{j} < 1$ and center $p$. Finally, recall that $\maR_{pj} =
\maV_{p(j-1)} \smallsetminus \maV_{pj} $. Let $R$ be any of the
regions $\Omega$, $\maR_{pj}$, or $\maV_{pn}$. Since the union of
these regions is $\PP$, it is enough to prove that
\begin{equation*}
    \|u - u_{I, \maT_n} \|_{\maK^1_1 (R \smallsetminus \maS )} \leq
    C \dim(S_n)^{-m/3}\, \|u\|_{\maK^{m+1}_{a+1} (R \smallsetminus
      \maS )},
\end{equation*}
for a constant $C$ independent of $R$ and $n$. The result will follow
by squaring all these inequalities and adding them up. In fact, since
$\dim(S_n)^{-m/3} = \maO( 2^{-nm})$, it is enough to prove
\begin{equation}\label{eq.R}
	\|u - u_{I, \maT_n} \|_{\maK^1_1 (R \smallsetminus \maS )}
        \leq C 2^{-nm}\, \|u\|_{\maK^{m+1}_{a+1} (R \smallsetminus
          \maS )},
\end{equation}
again for a constant $C$ independent of $R$ and $n$.

If $R = \Omega := \PP \smallsetminus \cup_p \maV_{p0}$, the estimate
in \eqref{eq.R} follows right away from Theorem
\ref{athm1.classical}. For the other estimates, recall that $0 < k \le
2^{-m/a}$, where $0 < a < \eta$ and $a\leq m$.  We next establish the
desired interpolation estimate on the region $R = \maR_{pj}$, for any
fixed $p \in \maS$ and $j = 1, 2, \ldots, n$. Let $\hat u(x) =
u(k^{j-1}x)$.  From Lemmas \ref{claim2} and \ref{claim1}, we have
\begin{equation*}
	\| u-u_{I, n} \|_{\maK^1_1(\maR_{pj})} = (k^{j-1})^{1/2} \|
        \hat u- \widehat{(u_{I, n}) }
        \|_{\maK^1_1(\maR_{p1})}=(k^{j-1})^{1/2} \| \hat u- \hat u_{I,
          n-j+1} \|_{\maK^1_1(\maR_{p1})}.
\end{equation*}
Since $\maK^m_a(\maR_{p1})$ is equivalent to $H^m(\maR_{p1})$, we can
apply Theorem \ref{athm1.classical} with $h =\maO( 2^{-(n-j+1)})$ to get
\begin{equation}\label{eqn.new6}
	\| u-u_{I, n} \|_{\maK^1_1(\maR_{pj})}\leq C
        (k^{j-1})^{1/2} 2^{-m(n-j+1)} \| \hat u
        \|_{\maK^{m+1}_{a+1}(\maR_{p1})}.
\end{equation}
Now applying Lemma \ref{claim2} to scale back again and using also
$k = 2^{-m/a}$, we get that the right hand side in
\eqref{eqn.new6}
\begin{multline*}
	C (k^{j-1})^{1/2} 2^{-m(n-j+1)} \| \hat u
        \|_{\maK^{m+1}_{a+1}(\maR_{p1})} = C (k^{j-1})^a
        2^{-m(n-j+1)} \| u \|_{\maK^{m+1}_{a+1}(\maR_{pj})} \\
        \le C 2^{-mn} \|u\|_{\maK^{m+1}_{a+1}(\maR_{pj})}.
\end{multline*}
This proves the estimate in \eqref{eq.R} for $R = \maR_{pj}$.

It remains to prove this estimate for $R = \maV_{pn}$.  For any
function $w$ on $\maV_{pn}$, we let $\hat w(x) = w(k^n x)$ be a
function on $\maV_{p}$.  Therefore, by Lemma \ref{claim2}
\begin{equation}\label{eqn.new5}
    \|u - u_{I,n} \|_{\maK^1_1 (\maV_{pn})}
    = (k^n)^{1/2} \|\widehat{u - u_{I,n}} \|_{\maK^1_1 (\maV_{p})}
\end{equation}
and by \ref{claim1} (which follows from the definition of the meshes
$\maT_k$ and from the fact that interpolation commutes with changes of
variables),
\begin{equation}\label{eqn.new4}
    (k^n)^{1/2} \|\widehat{u - u_{I,n}} \|_{\maK^1_1 (\maV_{p})}
  = (k^n)^{1/2} \|\hat u - \hat u_{I,0} \|_{\maK^1_1 (\maV_{p})}.
\end{equation}
Now let $\chi$ be a smooth cutoff function on $\maV_{p}$ such that
$\chi=0$ in a neighborhood of $p$ and $=1$ at every other node of
$\maV_{p}$.

Define $\hat v:=\hat u-\chi\hat u$. Then, by \eqref{eqn.new7},
\begin{eqnarray}
    (k^n)^{1/2} \|\hat u - \hat u_{I, 0} \|_{\maK^1_1
    (\maV_{p})}& =& (k^n)^{1/2} \|\hat v+\chi \hat u - \hat u_{I,
    0} \|_{\maK^1_1 (\maV_{p})} \nonumber\\
   & \leq& (k^n)^{1/2} (\|\hat v\|_{\maK^1_1
    (\maV_{p})}+\|\chi\hat u-\hat u_{I, 0}\|_{\maK^1_1
    (\maV_{p})})\nonumber\\
  &=& (k^n)^{1/2} (\|\hat v\|_{\maK^1_1 (\maV_{p})}+\|\chi\hat
  u-(\chi\hat u)_{I, 0}\|_{\maK^1_1 (\maV_{p})}). \label{eqn.new1}
\end{eqnarray}
Since $\chi$ vanishes in the neighborhood of $p$ we can consider
multiplication by $\chi$ as $\rho^{\infty}$ times a degree 0
b-operator. Thus it is a bounded operator on any weighted Sobolev
space. Thus
\begin{eqnarray}\label{eqn.new2}
  \|\hat v\|_{\maK^1_1(\maV_{p})}\leq \|\hat
  v\|_{\maK^m_1(\maV_{p})}\leq\|\hat
  u\|_{\maK^m_1(\maV_{p})}+\|\chi\hat u\|_{\maK^m_1(\maV_{p})} \leq
  C\|\hat u\|_{\maK^m_1(\maV_{p})},
\end{eqnarray}
where $C$ depends on $m$ and, through $\chi$, the nodes in the
triangulation.

Using \eqref{eqn.new5}, \eqref{eqn.new4}, \eqref{eqn.new1},
\eqref{eqn.new2}, Lemma \ref{claim3}, and Theorem \ref{athm1.classical},
we have
\begin{eqnarray*}
    \|u - u_{I,n}\|_{\maK^1_1 (\maV_{pn})} & \leq& C
    (k^n)^{1/2}(\|\hat u\|_{\maK^1_1(\maV_{p})} + \|\chi\hat u -
    (\chi\hat u)_{I, 0}\|_{\maK^1_1 (\maV_{p})}) \\
    & \leq & C(k^n)^{1/2}(\|\hat u\|_{\maK^1_1(\maV_{p})} +
    \|\chi\hat u\|_{H^{m+1} (\maV_{p})})\\
    & \leq & C(k^n)^{1/2}(\|\hat u\|_{\maK^1_1(\maV_{p})} +
    \|\hat u\|_{\maK_1^{m+1} (\maV_{p})})\\
    & \leq & C\| u\|_{\maK_1^{m+1}(\maV_{pn})}\leq Ck^{na}\|
    u\|_{\maK_{a+1}^{m+1}(\maV_{pn})}\leq
    C2^{-mn} \|u\|_{\maK_{a+1}^{m+1}(\maV_{pn})}.
\end{eqnarray*}

This proves the estimate of Equation \eqref{eq.R} for $R = V_{pn}$ and
completes the proof of Theorem \ref{athm1.gen}.
\end{proof}

\section{Applications to Finite Element Methods\label{sec.FEM}}

We can now turn to the proofs of the theorems stated in the introduction.
First, Theorem \ref{athm1.eig} follows from our general
approximation result, Theorem \ref{athm1.gen}, and the standard results
on approximations of eigenvalues and eigenvectors (eigenfunctions in
our case) discussed, for instance, in \cite{BabuOsborn1, BabuOsborn2, BabuOsborn3, BrambleOsborn,
  Osborn}.
More precisely, using the notation introduced in the introduction, 
we have the following. Let us denote by $E(\lambda)$
the eigenspace of $\Hk$ corresponding to the eigenvalue $\lambda$ and by
$E_1(\lambda) \subset E(\lambda)$, the subspace consisting of
vectors of length one. Then the following result is well known
(see for instance Equations (1.1) and (1.2) in \cite{BabuOsborn2}).
We state it only for our operator $\Hk$, although it is valid for more
general self-adjoint operators with compact resolvent.

\begin{theorem} 
\label{thm.eig.a}
There exists a constant $C>0$ with the following property. Let $V
\subset \maK_{1}^{1}(\TmS)$ be a finite dimensional subspace and $R :
\maK_{1}^{1}(\TmS) \to V$ the projection in the energy norm. Let
$w_{j, n} \in V$ be an eigenbasis of $R \Hk R$, namely $R \Hk Rw_{j,
  n}= R \Hk w_{j, n} = \lambda_{j,n} w_{j,n}$, with the
$\lambda_{j,n}$ arranged in increasing order in $j$. Then
\begin{equation*}
  |\lambda_j - \lambda_{j, n}| \le C \sup_{u \in E_1(\lambda)}
  \inf_{\chi \in V} \|u - \chi\|_{\maK_1^1(\TmS)}^2
\end{equation*}
and, furthermore
\begin{equation*}
  \|v_{j} - w_{j, n} \|_{\maK^1_1(\TmS)} \leq C \sup_{u \in
    E_1(\lambda)} \inf_{\chi \in V} \|u - \chi\|_{\maK_1^1(\TmS)},
\end{equation*}
for a suitable eigenvector $u_j \in E(\lambda_j)$.
\end{theorem}

The proof of Theorem \ref{athm1.eig} will then be obtained from Theorem
\ref{thm.eig.a} as follows.

\begin{proof} (of Theorem \ref{athm1.eig}).
We need to estimate $\sup_{u \in E_1(\lambda)} \inf_{\chi \in S_n} \|u
- \chi\|_{\maK_1^1}$. To this end, let us notice that any $u \in
E(\lambda) \subset \maK_1^1(\TmS)$ satisfies $(\mu + \Hk) u = (\mu +
\lambda) u$. Theorem \ref{theorem1} then gives
$\|u\|_{\maK_{a+1}^{m+1}} \le C_{m, \lambda} \|u\|_{\maK_{a-1}^{m-1}}$
for a suitably large $\mu$ that depends on $\lambda$ and $a < \eta$.
A bootstrap argument then gives for any $u \in E(\lambda)$ that
$\|u\|_{\maK_{a+1}^{m+1}} \le C'_{m, \lambda} \|u\|_{\maK_{1}^{1}}$.
Theorem \ref{athm1.gen} then gives for $u \in E_1(\lambda_j)$ (thus
$\|u\|_{\maK_{1}^{1}} = 1$), the following.
\begin{multline*}
  \sup_{u \in E_1(\lambda)} \inf_{\chi \in S_n} \|u -
  \chi\|_{\maK_1^1(\TmS)} \leq \sup_{u \in E_1(\lambda)} \|u - u_{I,
    \maT_n} \|_{\maK^1_1 (\TmS )} \\ 
  \leq C \sup_{u \in E_1(\lambda)} \dim(S_n)^{-m/3}
  \|u\|_{\maK_{a+1}^{m+1}(\TmS)} \leq c(m, \lambda_j)
  \dim(S_n)^{-m/3}.
\end{multline*} 
The proof of Theorem \ref{athm1.eig} is now complete.
\end{proof}

Next, the proof of Theorem \ref{athm1.fem} follows
from Theorem \ref{athm1.gen}, the Lax-Milgram Lemma and Cea's lemma.
We note some consequences of this theorem.

\begin{remark}\label{rk.opt}
First, in the case $f\in H^{m-1}(\TmS)$, by the estimate in
Equation \eqref{nnn1}, we have
\begin{eqnarray*}
\|v - v_{n} \|_{\maK^1_1(\TmS)} \leq C \dim(S_n)^{-m/3}
        \|f\|_{\maK^{m-1}_{a-1}(\TmS)}\leq C \dim(S_n)^{-m/3}
        \|f\|_{H^{m-1}(\TmS)},
\end{eqnarray*}
as long as the index in Theorem \ref{athm1.fem} is chosen such that
$0< a\leq 1$.
\end{remark}

As in the classical Finite Element Method, a duality argument yields
the following $L^2$-convergence result.
\begin{theorem}
In addition to the assumptions and notation in Theorem
\ref{athm1.fem}, assume that $0<a\leq 1$. Then the following $L^2$
estimate holds
\begin{eqnarray*}
  \|v - v_{n} \|_{L^2(\TT)} \leq C \dim(S_n)^{(-m-1)/3}
      {\|f\|_{H^{m-1}(\TT)}}.
\end{eqnarray*}
\end{theorem}

\begin{proof}
We sketch the proof by using the duality argument in weighted Sobolev
spaces. Consider the equation
\begin{eqnarray}\label{eqn.nnn1}
  (L+\Hk) w \hspace{0.09cm} = \hspace{0.05cm} v-v_n\quad
  \mbox{in} \quad \TT.
\end{eqnarray}
(So we use periodic boundary conditions on $\PP$.) {The
  definition of the Galerkin projection $v_n$ of $v$, Equation
  \eqref{eq.FEM.eq}, then gives}
\begin{eqnarray*}
  (v-v_n, v-v_n) = ((L+\Hk) w, v-v_n) = ((L+\Hk) (w-w_n),
  v-v_n),
\end{eqnarray*}
where $w_n$ is the finite element solution of Equation
\eqref{eqn.nnn1} on $\maT_n$.  {We also have
  $\|w\|_{\maK^2_{a+1}(\TmS)} \le C \|v - v_n\|_{L^2(\TT)}$ by
  Theorem \ref{theorem1}, since $v - v_n \in L^2(\TT) \subset
  \maK_{a-1}^{0}(\TmS)$.} Therefore,  applying Theorem \ref{athm1.fem}
 to $v-v_n \in L^2(\TT)$ and $m=1$, we have
\begin{eqnarray*}
  \|v-v_n\|_{L^2(\TT)} & \leq & {C} \|w-w_n\|_{\maK^1_1(\TT)}
  \|v-v_n\|_{\maK^1_1(\TT)}/\|v-v_n\|_{L^2(\TT)}\\
  & \leq & C \dim(S_n)^{-1/3} \|v-v_n\|_{\maK^1_1(\TT)} \leq C
  \dim(S_n)^{(-m-1)/3} \|f\|_{H^{m-1}(\TT)}.
\end{eqnarray*}
This completes the proof.
\end{proof}

\subsection{Condition number of the stiffness matrix}
 
It is important that the discrete system $S_n$ that we use
is well-conditioned for us to be able to realise the theoretical approximation
bounds in practice.  
Thus we need to additionally obtain upper and lower bounds on the eigenvalues of the stiffness
matrix that arises in calculation.

 Recall the standard
nodal basis function $\phi_j$ of the space $S_n := S(\maT_n, m)$. It
consists of functions that are equal to $1$ at one node and equal to
zero at all the other nodes. For convenience, we now instead consider
the rescaled bases $\varphi_j:=h_j^{-1/2}\phi_j$, where $h_j$ is the
diameter of the support patch for $\phi_j$. Then, we consider the
scaled stiffness matrix
\begin{equation}\label{eq.def.stiffn}
  A_n := \big(a(\varphi_i, \varphi_j)\big)
\end{equation}
from our graded finite element discretization \eqref{eq.FEM.eq}. In
practice, $A_n$ can be obtained from the usual stiffness matrix
$\big(a(\phi_i, \phi_j)\big)$ by a diagonal preconditioning
process. We point out that similar scaled matrices were considered in
\cite{BS89, Li12} for condition numbers of other Galerkin-based
methods.
  
For a symmetric matrix $A$, we shall denote by $\lambda_{max}(A)$ the
largest eigenvalue of $A$ and by $\lambda_{min}(A)$ the smallest
eigenvalue of $A$. Thus the spectrum of $A$ is contained in
$[\lambda_{min}(A), \lambda_{max}(A)]$, but is not contained in any
smaller interval. We first have the following estimates regarding properties of functions.

\begin{lemma}\label{thm.inverse}
Let $T_i$ be a tetrahedron in the mesh $\maT_n$ and let $\diam(T_i)$
denote the diameter of $T_i$. Then, for any $\mu_n\in S_n$ and $\mu\in
H^1(\Omega)$, there exists a constant $C>0$ independent of $n$,
$\mu_n$ and $\mu$, such that
\begin{eqnarray}
  & \|\mu_n\|_{H^1(T_i)}\leq C 
  {\diam(T_i)^{1/2}} \|\mu_n\|_{L^\infty(T_i)}\leq C
  \|\mu_n\|_{L^{6}(T_i)}, \label{eqn.inverse}\\
  & \|\mu\|_{L^6(\Omega)}\leq
  C\|\mu\|_{H^1(\Omega)}. \label{eqn.embedding}
\end{eqnarray}
\end{lemma}
Furthermore, writing $\mu_n=\sum c_j\varphi_j$, where $\varphi_j :=
h_j^{-1/2}\phi_j$ are rescaled basis functions, then
\begin{eqnarray}\label{eqn.norme}
  C^{-1/2} \sum_{j \in node(T_i)}\, c_j^2 \leq \, { \diam(T_{i}) }
  \|\mu_n\|^2_{L^{\infty}(T_i)} \leq \, C \sum_{j\in node(T_i)} c_j^2.
\end{eqnarray}
\begin{proof}
We shall show \eqref{eqn.inverse} and \eqref{eqn.norme} since
\eqref{eqn.embedding} is a standard result in \cite{GT77}. Recall that all
the tetrahedra $T_i$ belong to a finite class of shapes (or similarity
classes) in our graded triangulation. Thus, the bounded constant $C$
in \eqref{eqn.inverse} follows from the inverse estimates in
\cite{BrennerScott, Ciarlet78}.

As for \eqref{eqn.norme}, note $\mu_n = \sum c_i \varphi_i = \sum \bar
c_i \phi_i$. Based on the definition of the basis function $\varphi_i$
and of the graded mesh,
\begin{eqnarray}\label{eqn.eqn}
  C^{-1} \diam(T_i) ^{1/2}_i \bar c_i \leq c_i \leq C
  \diam(T_i)^{1/2}_i \bar c_i.
\end{eqnarray}  
On the reference tetrahedron $\hat T$, both $\|\hat v\|_{L^\infty}$
and $(\sum_{j\in node(\hat T)} \bar c_j^2)^{1/2}$ are norms for the
finite element function $\hat v|_{\hat T}$, where $\hat v$ is obtained by the
usual scaling process and the summation on $\bar c_j$ is for all the
nodes in $\hat T$. Based on equivalence of all norms  for a finite
dimensional space, we have
\begin{equation*}
  C(\sum_{j\in node(\hat T)} \bar c_j^2)^{1/2} \, \leq \, \|\hat
  v\|_{L^{\infty}(\hat T)} \, \leq \, C(\sum_{j\in node(\hat T)} \bar
  c_j^2)^{1/2}.
\end{equation*}
This, together with \eqref{eqn.eqn}, implies
\begin{equation*}
  C \sum_{j\in node(T_i)} c_j^2 \leq \diam(T_i)
  \|v\|^2_{L^{\infty}(T_i)}\leq C\sum_{j\in node(T_i)} c_j^2,
\end{equation*}
which completes the proof.
\end{proof}

Therefore, we have the following estimates on the eigenvalues of the
stiffness matrix.

\begin{lemma}\label{lem.upperbound}
Let $A_n$ be the stiffness matrix from the finite element
discretization {corresponding to the rescaled nodal basis $\varphi_j$
  of the space $S_n := S(\maT_n, m)$ in Equation
  \eqref{eq.def.stiffn}.} Then,
\begin{eqnarray*}
  \lambda_{max}(A_n)\leq M,
\end{eqnarray*}
where the constant $M$ is independent of the mesh level $n$. 
\end{lemma}

\begin{proof}
Let us fix the mesh level $n$. All the constants below will be independent of $n$. Let $\{T_i\}$ be the tetrahedra forming our mesh $\maT_n$. Let
$v \in {S_n}$ be arbitrary and write $v = \sum_jc_{j}\varphi_j$
and $\mathbf V:=(c_{j})$. Then, by Lemma 3.4 in \cite{HLNU1} we have
\begin{eqnarray*}
  \mathbf V^T A_n \mathbf V  =  a(v, v) \leq
  C\|v\|^2_{\maK_1^1(\PP)} \leq C \|v\|^2_{H^1(\PP)} \leq C
  \sum_i\|v\|^2_{H^1(T_i)}.
\end{eqnarray*}
By the inverse inequality \eqref{eqn.inverse} and the estimate
\eqref{eqn.norme}, we further have
\begin{eqnarray*}
   \mathbf V^T A_n \mathbf V \leq C \sum_i
   \diam(T_i) \|v\|^2_{L^\infty(T_i)} { \leq } C \sum_jc_{j}^2 \leq
   C \mathbf V^T \mathbf V,
\end{eqnarray*}
where $\diam(T_i)$ is the diameter of the tetrahedron $T_i$.  This
completes the proof.
\end{proof}


\begin{lemma}\label{lem.n3} 
We use the same notation as the one for Lemma \ref{lem.upperbound}.
The smallest eigenvalue of the stiffness matrix $A_n$,
\begin{equation*}
  \lambda_{min}(A_n) \geq C \dim(S_n)^{-2/3}.
\end{equation*}

\begin{proof} 
For any $v\in {S_n}$, we use the notation $v =
\sum_{j} c_{j} \varphi_j$, $\mathbf V := (c_{j})$, and $\diam(T_i)$
denotes the diameter of $T_i$, as before. In view of
\eqref{eqn.norme}, the inverse estimate \eqref{eqn.inverse},
H\"older's inequality, and the Sobolev embedding estimate
\eqref{eqn.embedding}, we then have
\begin{eqnarray*}
  \mathbf V^T \mathbf V & = & \sum_jc_{j}^2 \leq C\sum_i \diam(T_i)
  \|v\|^2_{L^\infty(T_i)} \leq C \sum_i\|v\|^2_{L^{6}(T_i)}\\
  & \leq & C \Big (\sum_i 1 \Big )^{\frac{2}{3}}\, \Big( \sum_i
  \|v\|^{6}_{L^{6}(T_i)} \Big )^{\frac{1}{3}} \leq C
  \dim(S_n)^{\frac{2}{3}} \|v\|^2_{L^6(\PP)} \\
  & \leq & C \dim(S_n)^{\frac{2}{3}}\, \|v\|_{H^1(\PP)}^2\leq
  C\dim(S_n)^{\frac{2}{3}}\, \mathbf V^TA \mathbf V.
\end{eqnarray*}
\end{proof}
\end{lemma}

Then, we have the estimate on the condition number.
\begin{theorem}\label{thm.main11}Let 
$A=(a(\varphi_i, \varphi_j))$ be the stiffness matrix. Then the
  condition number $\kappa(A)$ satisfies
\begin{eqnarray*}
  \kappa(A)\leq C\dim(S_n)^{2/3}.
\end{eqnarray*}
The constant $C$ depends on the finite element space, but not on
$\dim(S_n)$. 
\end{theorem}
\begin{proof}
Using $k(A) = \lambda_{max}(A)/\lambda_{min}(A)$, we obtain the
estimate by Lemmas \ref{lem.upperbound} and \ref{lem.n3}.
\end{proof}

\section{Numerical tests of the finite element method\label{sec.tests}}

We now present the numerical tests for the finite element solution
defined in \eqref{eqn.new8} approximating possibly singular solutions
to Equation \ref{eq.PDE.eq}.

To be more precise, suppose that our periodicity lattice is $2\ZZ^3$
and we choose our fundamental domain $\PP= [-1, 1]^3$ to be a cube of
side length 2.  We impose periodic boundary condition on the following
model problem
\begin{eqnarray}\label{eqn.axis1}
 (L + \Hk) v := (-\Delta +\delta\psi r^{-2} + L)v=1\quad
  {\rm{in}} \quad \Omega,
\end{eqnarray}
where $r=|x|$, $\delta>-1/4$, $L\geq 0$, and the cut-off
function $\psi:=e^{r_c^2/(r^4-r_c^2)+1}$ for $r^2\leq r_c$ and
$\psi=0$ for $r^2> r_c$; in the tests, we chose $r_c=0.25$.  Note that
if $\delta>0$, it is clear that the operator $L + \Hk$ is
positive on $\maK^1_1$ (see Theorem \ref{theorem1}). We use the $C^0$
linear finite element method on triangulations graded toward the
origin with grading ratio $k>0$ (Recall that $k=0.5$ corresponds to
the quasi-uniform refinement.)

To enforce the periodic boundary condition for the finite element
functions, we use meshes where all the boundary nodal points are
symmetric about the mid-plane between opposite faces of the cube. Any
set of the symmetric nodes will be associated to the same shape
function in the discretization. For example, nodes on edges of the
cube generally have three mirror images over two mid-planes (two
direct mirror images and the third is symmetric over the line of
intersection of these two mid-planes), and these four points are
associated to the same shape function. Consequently, the eight
vertices of the cube are associated to the same shape function through
symmetry. See Figure \ref{cube} for example.

\begin{figure}
\centerline{\includegraphics[scale=0.20]{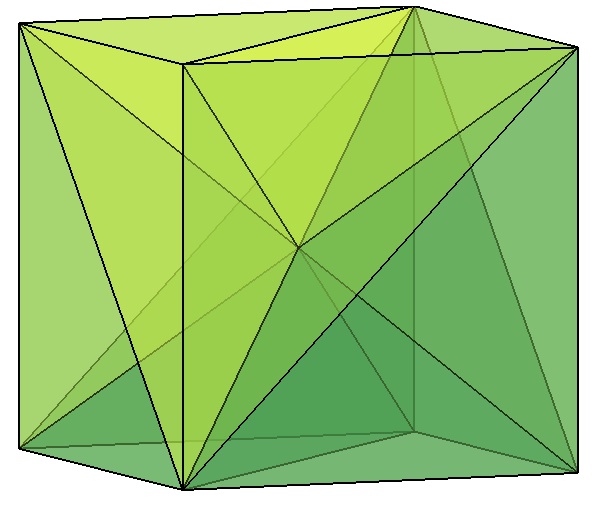} 
\hspace{3cm}\includegraphics[scale=0.20]{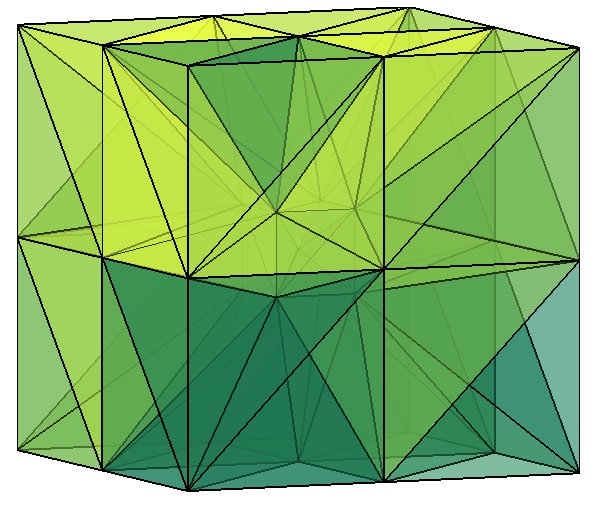}}
\caption{The initial mesh on the unit cube (left);
   the mesh after one $k$ refinement for the origin, $k=0.2$ (right).}
\label{cube}
\end{figure}

Our first tests are for Equation \eqref{eqn.axis1} with $\delta=4.0$
and $L=0$.  According to Theorem \ref{athm1.fem}, the optimal
rate of convergence for the Finite Element solution should be obtained
on triangulations with any $k\leq 0.5$, since
$\eta=\sqrt{1/4+4}>1$. The convergence rates $e$ associated to
triangulations with different values of $k$ are listed in Table
\ref{tab.5.3}. Starting from an initial triangulation, we compute the
rates based on the comparison of the numerical errors on
triangulations with consecutive $k$-refinements,
\begin{eqnarray}\label{e}
  e:=\log_2\frac{|v_{j-1}-v_{j}|_{\maK^1_1}}{|v_{j}-v_{j+1}|_{\maK_1^1}},
\end{eqnarray}
where $v_j$ is the finite element solution on the mesh after $j$
$k$-refinements. {Recall the dimension of the finite element space grows by a factor of 8 with one $k$-refinement. By Theorem \ref{athm1.fem}, for a sequence of  optimal meshes, the error $|v-v_j|_{\maK^1_1}$ is reduced by a factor of 2 for linear finite element approximations with each $k$-refinement.}Thus, $e\rightarrow 1$ implies that the optimal
rate of convergence in Theorem \ref{athm1.fem} is achieved.

Table \ref{tab.5.3} clearly shows that the convergence rates $e$
approach $1$ for all values of the grading parameter $k$. This is
in agreement with our theory that the optimal rates of convergence are
obtained for any triangulations with $k\leq 0.5$, since the
singularity in the solution is not strong enough to be detectable for
linear finite elements.

\begin{table}
\begin{tabular}{|l|l|}       \hline
  \emph{ $j\backslash e$ } & $k=0.1$ \hspace{0.5cm}
  $k=0.2$ \hspace{0.5cm} $k=0.3$ \hspace{0.5cm}
  $k=0.4$ \hspace{0.5cm} $k=0.5$\\ \hline
  \ \ 2 & 0.42 \hspace{1.1cm} 0.44 \hspace{1.1cm} 0.56 \hspace{1.1cm}
  0.33 \hspace{1.1cm} -0.20 \\ \hline
  \ \ 3 & 0.48 \hspace{1.1cm} 0.68 \hspace{1.1cm} 0.75 \hspace{1.1cm}
  0.79 \hspace{1.1cm} 0.70 \\ \hline
  \ \ 4 & 0.78 \hspace{1.1cm} 0.81 \hspace{1.1cm} 0.86 \hspace{1.1cm}
  0.88 \hspace{1.1cm} 0.85 \\ \hline
  \ \ 5 & 0.91 \hspace{1.1cm} 0.92 \hspace{1.1cm} 0.94 \hspace{1.1cm}
  0.95 \hspace{1.1cm} 0.93 \\ \hline
  \ \ 6 & 0.97 \hspace{1.1cm} 0.97 \hspace{1.1cm} 0.98 \hspace{1.1cm}
  0.99 \hspace{1.1cm} 0.98 \\ \hline
\end{tabular}
\caption{Convergence rates $e$ of finite element solutions solving
  equation \eqref{eqn.axis1} with $\delta=4.0$ and $L=0$ on
  different graded tetrahedra.}\label{tab.5.3}
\end{table}

In the second test, we implemented our method solving equation
\eqref{eqn.axis1} with $\delta=0.6$, $L=0$ and summarize the
results in Table \ref{tab.5.4}. Based on the upper bound
$\eta=\sqrt{1/4+0.6}$ given in Theorem \ref{athm1.fem}, we expect the
optimal rate of convergence for the numerical solution as long as the
grading parameter $k < 2^{-1/\eta}\approx0.47$. The convergence
rates in Table \ref{tab.5.4} tend to $1$ when $k\leq 0.4$, which
implies the optimality of our finite element approximation on these
meshes. However, when $k=0.5$, the convergence rate is far less
than $1$ and there is a large gap between the rates corresponding to
$k=0.4$ and $k=0.5$. This further confirms our theory that
the upper bound of the suitable range of $k$ for an optimal
finite element approximation lies in $(0.4, 0.5)$.

\begin{table}
\begin{tabular}{|l|l|}       \hline
  \emph{ $j\backslash e$ } & $k=0.1$ \hspace{0.5cm}
  $k=0.2$ \hspace{0.5cm} $k=0.3$ \hspace{0.5cm}
  $k=0.4$ \hspace{0.5cm} $k=0.5$\\ \hline
  \ \ 2 & 0.20 \hspace{1.1cm} 0.30 \hspace{1.1cm} 0.33 \hspace{1.1cm}
  0.11 \hspace{1.1cm} -0.03 \\ \hline
  \ \ 3 & 0.54 \hspace{1.1cm} 0.66 \hspace{1.1cm} 0.69 \hspace{1.1cm}
  0.61 \hspace{1.1cm} 0.39 \\ \hline
  \ \ 4 & 0.74 \hspace{1.1cm} 0.81 \hspace{1.1cm} 0.83 \hspace{1.1cm}
  0.77 \hspace{1.1cm} 0.60 \\ \hline
  \ \ 5 & 0.88 \hspace{1.1cm} 0.91 \hspace{1.1cm} 0.92 \hspace{1.1cm}
  0.87 \hspace{1.1cm} 0.72 \\ \hline
  \ \ 6 & 0.95 \hspace{1.1cm} 0.97 \hspace{1.1cm} 0.98 \hspace{1.1cm}
  0.92 \hspace{1.1cm} 0.79 \\ \hline
\end{tabular}
\caption{Convergence rates $e$ of finite element solutions solving
  equation \eqref{eqn.axis1} with $\delta=0.6$ and $L=0$ on different
  graded tetrahedra.}\label{tab.5.4}
\end{table}

The third tests are for negative potentials in equation
\eqref{eqn.axis1}, where we set $\delta=-0.1$ and $L=20$ to satisfy
the positivity requirement in Theorem \ref{theorem1}.  Our theoretical
results indicate that the singularity in the solution due to the
singular potential is stronger in this case and the optimal rate can
be achieved only if the grading parameter
$k<2^{-1/\sqrt{1/4-0.1}}\approx 0.167$. Because of the limitation of
the computation power, we only display the convergence results up to
the 7th refinement for various graded parameters $k$ in Table
\ref{tab.5.5}. We, however, still see the trend that appropriate
gradings improve the convergence rate as predicted in Theorem
\ref{athm1.fem}. When $k$ is close to the optimal value $0.167$ (i.e.,
$k=0.1$ and $0.2$), we have remarkable improvements.  In particular,
for $k=0.1$, based on Table \ref{tab.5.5}, we expect that the optimal
rate occurs with further refinements.

  We have also implemented the method on graded meshes
for the eigenvalue problem associated with equation \eqref{eqn.axis1},
especially on the computation of the first eigenvalues. Namely,
\begin{equation*}
 {\Hk} u:=(-\Delta +\delta\psi r^{-2})u=\lambda_1 u
\end{equation*}
on the unit cube, where $\lambda_1$ is the first eigenvalue of the
operator. Depending on the choice of $\delta$, the convergence rates
for the numerical eigenvalues on graded meshes are roughly twice the
rates for the numerical solutions of equation \eqref{eqn.axis1} (see
Tables \ref{tab.5.3}, \ref{tab.5.4}, and \ref{tab.5.5}), and present
similar trends for different gradings.

\begin{table}
\begin{tabular}{|l|l|}       \hline
  \emph{ $j\backslash e$ } & $k=0.1$ \hspace{0.5cm}
  $k=0.2$ \hspace{0.5cm} $k=0.3$ \hspace{0.5cm}
  $k=0.4$ \hspace{0.5cm} $k=0.5$\\ \hline
  \ \ 2 & -0.10 \hspace{.95cm} -0.05 \hspace{.95cm} -0.09 \hspace{.95cm}
  -0.16 \hspace{.95cm} -0.03 \\ \hline
  \ \ 3 & 0.32 \hspace{1.1cm} 0.37 \hspace{1.1cm} 0.30 \hspace{1.1cm}
  0.19 \hspace{1.1cm} 0.07 \\ \hline
  \ \ 4 & 0.51 \hspace{1.1cm} 0.52 \hspace{1.1cm} 0.44 \hspace{1.1cm}
  0.32 \hspace{1.1cm} 0.18 \\ \hline
  \ \ 5 & 0.67 \hspace{1.1cm} 0.64 \hspace{1.1cm} 0.53 \hspace{1.1cm}
  0.40 \hspace{1.1cm} 0.26 \\ \hline
  \ \ 6 & 0.80 \hspace{1.1cm} 0.72 \hspace{1.1cm} 0.59 \hspace{1.1cm}
  0.45 \hspace{1.1cm} 0.32 \\ \hline
\end{tabular}
\caption{Convergence rates $e$ of finite element solutions solving
  equation \eqref{eqn.axis1} with $\delta=-0.1$ and {$L=20$} on
  different graded tetrahedra.}\label{tab.5.5}
\end{table}

All our numerical tests (Tables \ref{tab.5.3},\ref{tab.5.4}
\ref{tab.5.5}, and corresponding eigenvalue computations) convincingly
verify Theorem \ref{athm1.eig} by comparing the rates of convergence
for different singular potentials on different graded triangulations
for the model operator in \eqref{eqn.axis1}. The theoretical upper
bounds $2^{-1/\eta}$ of the optimal range for the grading parameter
$k$ are also clearly demonstrated in these numerical results. In these
tests, the initial triangulation of the unit cube consists of 12
tetrahedra and we consecutively refine the mesh using the
$k$-refinements up to level 7 that includes $12\times 8^7\approx
2.5\times 10^7$ tetrahedra and roughly 4.2 million unknowns.
Numerical experiments show that the condition numbers of our discrete
systems grow by a factor of 4 for consecutive refinements, regardless
of the value of $k$, which resembles the estimates given in
\cite{BS89} for the Laplace operator. However, the values of $k$
affect the magnitude of the condition numbers. In general, smaller $k$
leads to bad shapes for the tetrahedra and therefore results in larger
condition numbers. The preconditioned conjugate gradient (PCG) method
was used as the numerical solver for the discrete systems.


\def\cprime{$'$} \def\ocirc#1{\ifmmode\setbox0=\hbox{$#1$}\dimen0=\ht0
  \advance\dimen0 by1pt\rlap{\hbox to\wd0{\hss\raise\dimen0
  \hbox{\hskip.2em$\scriptscriptstyle\circ$}\hss}}#1\else {\accent"17 #1}\fi}

\end{document}